\documentclass{article}

\usepackage{booktabs,multirow,array}
\usepackage{amsmath}
\usepackage{amsfonts}
\usepackage{amssymb}
\usepackage{amsthm}
\usepackage{units}
\usepackage{graphicx}
\usepackage{comment}
\usepackage{todonotes}
\usepackage{xcolor}
\usepackage{tikz}
\usetikzlibrary{shadows,arrows}
\usepackage{subcaption}
\usepackage{wrapfig}
\usepackage[T1]{fontenc} 

\usepackage{placeins}

\usepackage{hyperref}

\usepackage{listings}
\usepackage{mathtools}

\usepackage{thmtools}
\usepackage{thm-restate}
\usepackage[algoruled, algo2e, vlined, titlenotnumbered]{algorithm2e}

\newcommand\numberthis{\addtocounter{equation}{1}\tag{\theequation}}

\newcommand{\R}{\mathbb{R}}
\newcommand{\N}{\mathbb{N}}
\newcommand{\Zn}{\mathbb{Z}^n}

\newcommand{\cN}{\mathcal{N}}
\newcommand{\cD}{\mathcal{D}}
\newcommand{\NP}{\mbox{\slshape NP}}
\newcommand{\cF}{\mathcal{F}}
\renewcommand{\mid}{:}
\theoremstyle{plain}

\DeclareMathOperator*{\argmax}{arg\!\max}

\newtheorem{theorem}{Theorem}
\newtheorem{lemma}[theorem]{Lemma}
\newtheorem{corollary}[theorem]{Corrolary}

\theoremstyle{definition}
\newtheorem{defn}[theorem]{Definition}

 \author{Corinna Gottschalk \and Britta Peis} 
\title{Submodular Function Maximization over Distributive and Integer Lattices}

\begin{document}
\maketitle

\begin{abstract}
\noindent
The problem of maximizing non-negative submodular functions has been studied extensively in the last few years. 
However, most papers consider submodular set functions. 
Recently, several advances have been made for submodular functions on the integer lattice. 
As a direct generalization of submodular set functions, a function $f: \{0, \ldots, C\}^n \rightarrow \mathbb{R}_+$ is submodular, if
$f(x) + f(y) \geq f(x \land y) + f(x \lor y)$ for all $x,y \in  \{0, \ldots, C\}^n$ where $\land$ and $\lor$ denote 
element-wise minimum and maximum. 
The objective is finding a vector $x$ maximizing $f(x)$. In this paper, we present a deterministic $\frac{1}{3}$-approximation
using a framework inspired by \cite{DoubleGreedy}.
Moreover, we show that the analysis is tight and that other ideas used for approximating set functions cannot easily be extended. 

In contrast to set functions, submodularity on the integer lattice does not 
imply the so-called diminishing returns property. Assuming this property, it was shown that many results for set functions can
also be obtained for the integer lattice. 
In this paper, we consider a further generalization. Instead of the integer lattice, we consider a distributive lattice as the function domain
and assume the diminishing returns (DR) property.
On the one hand, we show that some approximation algorithms match the set functions setting. In particular, we can obtain a
$\nicefrac{1}{2}$-approximation for unconstrained maximization, a $(1-\nicefrac{1}{e})$-approximation for monotone functions under a cardinality constraint and 
a $\nicefrac{1}{2}$-approximation for a poset matroid constraint. 
On the other hand, for a knapsack constraint, the problem becomes significantly harder: 
even for monotone DR-submodular functions, we show that there is no $2^{(\log (n^{1/2} - 1))^\delta - 1}$-approximation 
for every $\delta > 0$ under the assumption that $3-SAT$ cannot be solved in time $2^{n^{3/4 + \epsilon}}$. 
\end{abstract}


\section{Introduction}
Recall that a set function $f : 2^\mathcal{N} \rightarrow \mathbb{R}$ is called submodular if $
f(U) + f(W) \geq f(U\cap W) + f(U\cup W)$ for all $U, W \subseteq \mathcal{N}$. 
Optimization problems with submodular objective functions have received a lot of attention in recent years.
Submodular objectives are motivated by the principle of economy of scale, and thus find many applications in real-world problems.
Moreover, submodular functions play a major role in combinatorial optimization.
Several combinatorial optimization problems have some underlying submodular structure, for example, 
cuts in graphs and hypergraphs, or rank functions of matroids.

As a breakthrough result, the problem to find a subset $S\subseteq \cN$ minimizing a submodular function $f$
has been shown to be solvable in strongly polynomial time in \cite{SubmodularMin1}.
In contrast, the corresponding maximization problem
\begin{equation}\label{USM}
\max \{f(S)\mid S\subseteq \cN\}
\end{equation}
for a nonnegative submodular function $f$ is easily seen to be \NP-hard, as it contains,  for example,
\textsc{max cut} as a special case.
We refer to \eqref{USM} as \textsc{unconstrained submodular maximization (USM)}.

Recently, a generalization to submodular functions on a bounded subset of the integer lattice $\Zn$ has also been investigated.
For $x, y \in \Zn$ let $(x \lor y)_e$ denote $\max\{x_e, y_e\}$ and $(x \land y)_e$ denote $\min\{x_e, y_e\}$ 
for $e \in \{1, \ldots, n\}$. 
Then, a function $f: D \rightarrow \R$  on a finite set $D$ of the form 
$D = \{x \in \mathbb{Z}^n |l_i \leq x_i \leq u_i \ \forall i \in \{1, \ldots, n\}\}$
is called submodular if 
$$f(x) + f(y) \geq f(x \lor y) + f(x \land y) \ \forall x, y \in D.$$ 
Clearly, this captures submodular set functions since vectors with entries $0$ and $1$ can be seen as incidence vectors of 
sets and in that case $\land$ and $\lor$ correspond to intersection and union. $D$ is called \emph{bounded integer lattice}.
Submodular functions on the integer lattice have been studied before, for example, in discrete convex analysis, 
 $L^\natural$-convex functions are of this type. We provide more details in ``Related Work''.  

 Submodularity for set functions is equivalent to the diminishing returns (DR) property: 
 For $S\subset T, e \not \in T$ we have $f(S + e) - f(S) \ge f(T + e) - f(T)$, a very natural property in many problems.
 Submodularity on the integer lattice, however, is a weaker property in the sense that Diminishing Returns imply submodularity, but not the other way round.
 Soma and Yoshida \cite{SomaYoshi16} have investigated the setting of DR-submodular functions on the integer lattice. We will use the prefix ``DR'' to describe
 problems that consider DR-submodular functions. 
 
In addition to the integer lattice, we can also consider a more general structure: distributive lattices, that is, lattices where the meet- and join-operation 
 fulfill distributivity. We will provide definitions in Section \ref{sec:preliminaries}. 
 For now, it is important, that due to Birkhoff's theorem, we can represent
 any distributive lattice as the set of all ideals on some poset $P$ and vice versa. 
 Form now on, we always assume that distributive lattices $\cD(P)$ are (implicitly) given by this poset.   

Maximizing a submodular function subject to some side constraints is also a well-understood problem. 
Typically, these constraints are matroid or knapsack constraints, with the special case of a cardinality constraint, which can be seen as a matroid constraint for 
the so called uniform matroid. 
In contrast to the unconstrained setting, considering monotone functions still yields interesting problems here. 

There are corresponding generalizations of matroids for both types of lattices presented above: 
On the integer lattice, we can consider integer polymatroids, on distributive lattices, there are poset-matroids. 
Hence, the problem of maximizing a submodular function subject to these constraints can be generalized as well. 
 
\subsection{Our results}
In this work, we present approximation algorithms for several constrained and unconstrained submodular maximization problems, 
in particular, we consider the following. 

\emph{USM on the integer lattice.}
Given a submodular function $f: D \rightarrow \R_+$ on a bounded integer lattice $D$, 
we show how to find a pseudopolynomial $\frac{1}{3}$-approximation. 
Moreover, we show that the analysis is tight and discuss some difficulties for obtaining 
better approximation guarantees using existing techniques that have been applied successfully for set functions ($C = 1$). 
Note that the one-dimensional case where $D = \{0, 1, \ldots, C\}$ already illustrates that obtaining an exact polynomial algorithm 
is not possible, since any function on $D$ is submodular.  
A preliminary version of these results was published in \cite{mySubmax}.  

\emph{DR-USM on the distributive lattice.}
Given a DR-submodular function $f: \cD(P) \rightarrow \R_+$ on a distributive lattice $\cD(P)$, we show how to generalize the algorithm of 
Buchbinder et al.\ \cite{DoubleGreedy} for set functions to obtain a randomized $\frac{1}{2}$-approximation. 

\emph{Cardinality and Poset Matroid Constraints on the distributive lattice.}
Given a monotone DR-submodular function $f: \cD(P) \rightarrow \R_+$ on a distributive lattice $\cD(P)$, 
we obtain a $(1 - \nicefrac{1}{e})$-approximation 
for maximization under a cardinality constraint and a $\frac{1}{2}$-approximation under a poset matroid constraint.
 
 
\emph{Hardness on the distributive lattice.}
While the above results might suggest that submodular maximization on the distributive lattice is not 
fundamentally different from the integer lattice or even from set functions, this surprising result shows otherwise. 
Maximizing a monotone, DR-submodular function under a knapsack constraint is well understood and approximable in the previously mentioned cases.  
However, we show, that for the distributive lattice, there is no $2^{(\log (n^{1/2} - 1))^\delta - 1}$-approximation for every $\delta > 0$ 
under the assumption that $3-SAT$ cannot be solved in time $2^{n^{3/4 + \epsilon}}$.

\subsection{Related Work}\label{sec:related-results}
A comprehensive study on USM has been done by Feige, Mirrokni and Vondr\'ak in \cite{LocalSearch} 
who provide and analyze several constant approximation algorithms for USM.
In particular, they present a simple Local Search algorithm that yields a $\frac{1}{3}$-approximation.
Using a noisy version of $f$ as objective function, they could improve the performance guarantee to $\frac{2}{5}$.
Feige et al.\ in \cite{LocalSearch} also showed that we cannot hope for a performance guarantee lower than $\frac{1}{2}$.
They could prove that any $\frac{1}{2} + \varepsilon$ -approximation would require an 
exponential number of queries to the oracle. For symmetric submodular functions, they already show that the ratio is tight.
Subsequently, Oveis Gharan and Vondr\'ak \cite{USMImprovement1} and Feldmann, Naor and Schwartz \cite{USMImprovement2} 
improved the approximation ratio to 0.41 and 0.42 respectively and finally Buchbinder, Feldmann, Naor 
and Schwartz closed the gap and gave a randomzied $\frac{1}{2}$-approximation for USM in \cite{DoubleGreedy}.
They also present a deterministic $\frac{1}{3}$-approximation using a similar idea. 
Both algorithms use a ``Double Greedy'' framework that starts with two different sets and, 
for a fixed order of the elements, decides in each step which of the two sets should be modified using the given element.
Later, Buchbinder and Feldman showed in \cite{DerandomizedDoubleGreedy} how to 
derandomize this algorithm to obtain a deterministic $\frac{1}{2}$-approximation.

For several special types of submodular functions,
better approximation ratios can be proved.
For example, Goemans and Williamson \cite{MaxCut1} 
provide an $0.878$-approximation for \textsc{max cut} and a $0.796$-approximation for \textsc{max dicut} based on semidefinite programming and 
Ageev and Sviridenko give a $0.828$-approximation for \textsc{maximum facility location} \cite{FacLoc1}. 

%
Furthermore, USM has applications in marketing in social networks and revenue maximization \cite{SocialNetworks}. 
Submodular function maximization also plays a crucial role in algorithmic game theory. 
For example, Dugmi et al.\ use USM as well as constrained submodular maximization for analyzing auctions in \cite{GameTheory1} and 
Schulz and Uhan use USM as a subroutine to calculate the \emph{least core} in cooperative games \cite{GameTheory2}.

Since vectors with entries $0$ and $1$ can be seen as incidence vectors of sets, 
submodular functions on the integer lattice can be seen as a direct generalization of submodular set functions and in particular, all hardness results still apply. 
Moreover, the integer lattice can be interpreted as a specific distributive lattice that consists of disjoint chains. 
However, we remark that this correspondence is only pseudopolynomial, i.e.\ we need a pseudopolynomial number of 
elements to represent the integer lattice as a distributive lattice.  
Moreover, Soma and Yoshida observed that, for a DR-submodular function on the integer lattice, 
there is an equivalent submodular function on a Boolean lattice whose size is pseudopolynomially bounded (by $n$ and the upper bound for the components).
We refer to a talk given at a Hausdorff Trimester Program which is also publicly available \cite{SomaBonnTalk}.
In particular, any approximation algorithm for submodular set functions yields a pseudopolynomial approximation 
of the same guarantee for DR-submodular functions on the integer lattice, where the running time depends on the bound $C$. 
Thus, assuming DR-submodularity, the challenge on the integer lattice is developing polynomial algorithms.
In particular, there is a polynomial $\frac{1}{2}$-approximation for DR-USM on the integer lattice \cite{SomaDR-SMBIL}. 
However, our hardness result suggests that no such easy reduction to submodular set functions exists for distributive lattices. 

In discrete convex analysis, several special cases of submodular functions on the integer lattice have been investigated. 
For example, $L^\natural$-convex and $M^\natural$-concave functions are submodular 
and $M ^\natural$-concave functions can be maximized in polynomial time \cite{Murota}.

As mentioned before, the problem of minimizing a submodular function is solvable in polynomial time. 
This is not only true for submodular set functions, but even for submodular functions on distributive lattices, 
or more precisely on ring families \cite{FujiDistrLattice}. On the other hand, nothing is known for maximization 
on the distributive lattice. 

A considerable amount of work deals with maximizing both monotone and non-monotone submodular functions subject to some side constraints. 
Most of these papers consider submodular set functions. The side constraints that are typically investigated are matroid and knapsack constraints 
and combinations of these. 
The first result in this area goes back to Fisher, Nemhauser and Wolsey, who considered monotone set functions and 
presented a $(1-\nicefrac{1}{e})$-approximation for a cardinality constraint \cite{FisherGreedyI} 
and a $\nicefrac{1}{P + 1}$-approximation for $P$ matroid constraints \cite{FisherGreedyII}. 
Calinescu et al.\ improved this to $(1 - \nicefrac{1}{e})$ for one matroid constraint \cite{MonotoneMatroid1}. 
Nemhauser and Wolsey \cite{NemhauserHardnessSM} showed that any 
approximation better than $(1 - \nicefrac{1}{e})$ requires an exponential number of queries if $f$ is accessed via a value oracle even for the case of a
cardinality constraint and monotone functions. 
Moreover, Feige showed that Max-$k$-Cover, which is a special case of maximizing a monotone submodular function subject to a cardinality constraint is hard
to approximate to within a factor of $(1 - \frac{1}{e} - \epsilon)$ for any fixed $\epsilon$ unless $P = NP$. In contrast to the other complexity result, 
the submodular function is given by a closed form and not just accessed by an oracle.   

For monotone set functions, Sviridenko \cite{SviridenkoKnapsack} provided the first $(1-\nicefrac{1}{e})$-approximation for a knapsack constraint. 
Since a cardinality constraint can also be interpreted as a knapsack constraint, the same hardness result applies.  

Inaba et al.\ \cite{Inaba:KnapsackIL} first considered constrained problems on the integer lattice and 
showed how to obtain a pseudopolynomial $(1 - \nicefrac{1}{e})$-approximation for the monotone knapsack case.
In \cite{SomaYoshi16}, Soma and Yoshida presented $(1 - \nicefrac{1}{e} - \epsilon)$-approximations with 
polynomial running times depending on $\frac{1}{\epsilon}$ for DR-submodular monotone functions and a cardinality, polymatroid or knapsack constraint.
Furthermore, they achieve the same approximation guarantee 
for submodular monotone functions and a cardinality constraint using an algorithm whose running time depends 
logarithmically on the ratio of the maximum function value and the minimum positive increase. 


In \cite{ContentionResolution}, Chekuri, Vondr\'ak and Zenklusen consider the very general problem 
of maximizing  a non-monotone submodular function on a down-monotone family of sets, 
which in particular contains the previously mentioned constraints. 
They also provide a concise overview over previous results in non-monotone maximization and we refer to
their paper for details. 
To the best of our knowledge, no results for non-monotone constrained maximization on the integer lattice exist. 

While our main interest in SMBIL is of a theoretical nature with the ultimate goal to gain a better understanding of submodular 
function maximization, there are applications where submodular functions on the integer lattice play a crucial role.
For example, Bolandnazar et al.\ (\cite{SubmodSupplyChains}) showed that Assemble-to-Order Systems for supply chain management
can be optimized by  minimizing a submodular function on the integer lattice. 
\cite{Inaba:KnapsackIL} discusses several applications for submodular maximization on the integer lattice, e.g.\ budget allocation and sensor placement. 
There is also an alternative way to interpret the bounded integer lattice: Let $x \in \mathbb{Z}_+^n$ denote the incidence vector of a multiset
where entries in $x$ specify the multiplicity of individual elements. Then SMBIL can be seen as maximizing a submodular function
on multisets containing at most $C$ copies of each element which gives rise to possible applications. This illustrates for example how sensor placement can be 
seen as an application of SMBIL by allowing the possibility to place multiple sensors at the same location. 

In Figure \ref{fig:SM-overview}, we provide a graphic summary of all problems considered in this paper and related problems.  
An arrow from one box $A$ to box $B$ indicates that problem  $A$ generalizes $B$. 
We denote distributive, integer and Boolean lattices by DL, IL and Boolean respectively, where Boolean is equivalent to the the well-studied case of set functions. 
The diminishing returns property is denoted by DR and we remind the reader that every submodular set function is  DR. 
This is not given on integer or distributive lattices, 
where DR-submodular functions are a subclass of submodular functions. 
Those cases where we only have pseudopolynomial algorithms are marked as such by the additional tag "pseudo". 
Due to results on hardness of approximation for USM and ($\le k$, Boolean), all approximation guarantees except for (poset matroid, DL, DR) and (SMBIL) are tight (up to $\epsilon$).

The paper is organized as follows: The necessary terminology is introduced in Section \ref{sec:preliminaries}. Then, we first investigate unconstrained maximization
on the integer lattice in Section \ref{sec:SMBIL-Algo}. In Section \ref{sec:SM-DR-DL-algos} we present algorithms for several 
maximization problems for DR-submodular functions on the distributive lattice and in Section \ref{sec:SM-DR-DL-hardness}, we present a
strong hardness result. 

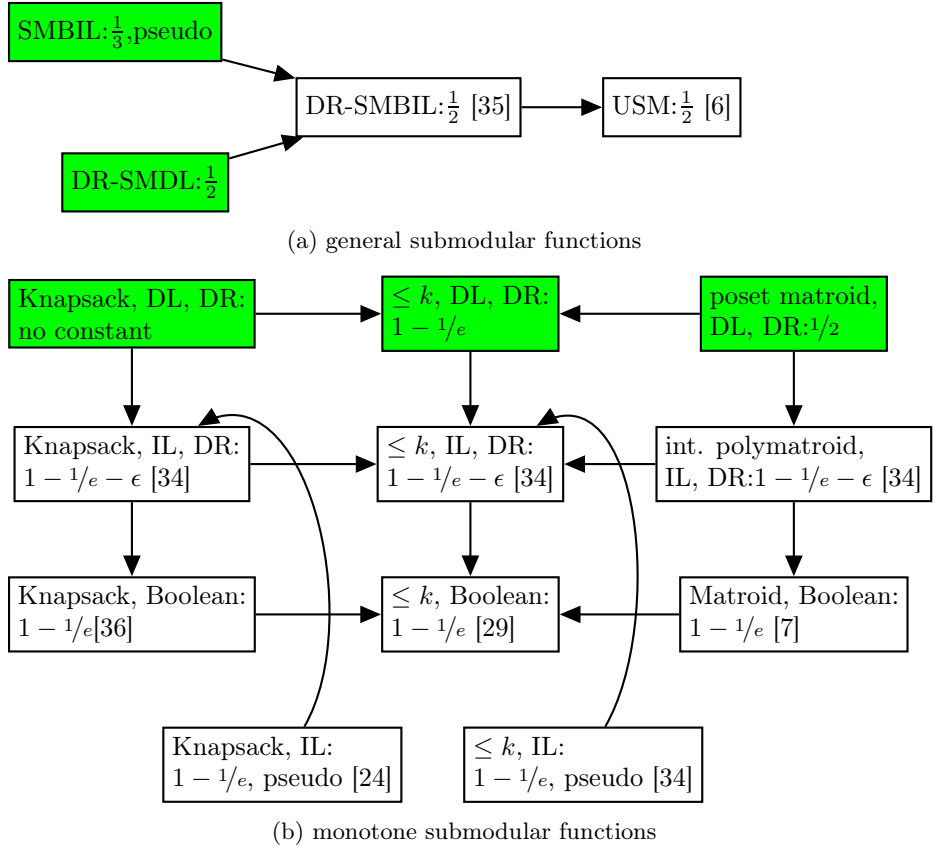
\begin{figure}[tb]
  
 \tikzset{%
  block/.style    = {draw, thick, rectangle, minimum height = 2.2em,
    minimum width = 3em}
}

\begin{subfigure}{\textwidth}
\begin{tikzpicture}[auto, thick, node distance=2cm, >=triangle 45]

 \draw node at (6.5, 0)[block] (USM) {USM:{$\frac{1}{2}$ \cite{DoubleGreedy}}};
  \draw node at (3, 0)[block] (DR-SMBIL) {DR-SMBIL:$\frac{1}{2}$ \cite{SomaDR-SMBIL}};
  \draw node at (-0.9, 1)[block, fill = green] (SMBIL) {SMBIL:$\frac{1}{3}$,pseudo};
  \draw node at (-0.5, -1)[block, fill = green] (DR-SMDL) {DR-SMDL:$\frac{1}{2}$};
  \draw[->](DR-SMBIL) -- (USM);
  \draw[->](SMBIL) -- (DR-SMBIL);
 \draw[->](DR-SMDL) -- (DR-SMBIL);
 
 \end{tikzpicture}
 \caption{general submodular functions}
 \end{subfigure}

 \begin{subfigure}{\textwidth}
 \begin{tikzpicture}[auto, thick, node distance=2cm, >=triangle 45]

 		\draw [white] (0,4.6) circle (1mm);  
       \draw node at (4.5, 0)[block, align = left] (CardB) {$\le k$, Boolean:\\ $1 - \nicefrac{1}{e}$ \cite{FisherGreedyI}};
       \draw node at (8.8, 0)[block, align = left] (MatB) {Matroid, Boolean:\\ $1 - \nicefrac{1}{e}$ \cite{MonotoneMatroid1}};
       \draw node at (0, 0)[block, align = left] (KnapB) {Knapsack, Boolean:\\ $1 - \nicefrac{1}{e}$\cite{SviridenkoKnapsack}};
       
       \draw node at (4.5, 2)[block, align = left] (CardILD) {$\le k$, IL, DR:\\ $1 - \nicefrac{1}{e} - \epsilon$ \cite{SomaYoshi16}};
       \draw node at (8.8, 2)[block, align = left] (MatILD) {int. polymatroid,\\ IL, DR:$1 - \nicefrac{1}{e}- \epsilon$ \cite{SomaYoshi16}};
       \draw node at (0, 2)[block, align = left] (KnapILD) {Knapsack, IL, DR:\\ $1 - \nicefrac{1}{e}- \epsilon$ \cite{SomaYoshi16}};
       
       \draw node at (4.5, 4)[block, align = left, fill = green] (CardDL) {$\le k$, DL, DR:\\ $1 - \nicefrac{1}{e}$};
       \draw node at (8.8, 4)[block, align = left, fill = green] (MatDL) {poset matroid,\\ DL, DR:$\nicefrac{1}{2}$};
       \draw node at (0, 4)[block, align = left, fill = green] (KnapDL) {Knapsack, DL, DR:\\ no constant};
       
       \draw node at (2, -2)[block, align = left] (KnapIL) {Knapsack, IL:\\ $1 - \nicefrac{1}{e}$, pseudo \cite{Inaba:KnapsackIL}};
       \draw node at (6, -2)[block, align = left] (CardIL) {$\le k$, IL:\\ $1 - \nicefrac{1}{e}$, pseudo \cite{SomaYoshi16}};
       
       \draw[->](MatB) -- (CardB);
       \draw[->](KnapB) -- (CardB);
       \draw[->](CardILD) -- (CardB);
       
        \draw[->](KnapILD) -- (CardILD);
         \draw[->](MatILD) -- (CardILD);
          \draw[->](CardDL) -- (CardILD);
          
           \draw[->](MatDL) -- (CardDL);
              \draw[->](KnapDL) -- (CardDL);
              
              \draw[->](KnapDL) -- (KnapILD);
              \draw[->](KnapILD) -- (KnapB);
              \draw[->](MatDL) -- (MatILD);
              \draw[->](MatILD) -- (MatB);

	\draw[ ->] (KnapIL) to [out=60,in=30](KnapILD);
	\draw[->] (CardIL) to [out=60,in=30] (CardILD);
 \end{tikzpicture}
 \caption{monotone submodular functions}
 \end{subfigure}

\caption{Overview of all problems, those considered in this paper are shaded. 
}\label{fig:SM-overview}
\end{figure}


\section{Preliminaries}\label{sec:preliminaries}
For submodular set functions, it is \NP-hard to decide
whether there exists a set $S$ such that $f(S) > 0$, as mentioned by Feige et al.\ in \cite{LocalSearch}. 
Since our main goal is finding good approximation algorithms, 
throughout this paper, we assume that all submodular functions are nonnegative and given by a value oracle.

For a vector $x \in \Zn$ we denote by $(x|x_j = k)$ the vector where all 
but the entry $x_j$ remain the same and $x_j$ is set to k.
As usual, $[n]$ denotes the set $\{0,1, \ldots, n\}$ and $e^i \in \{0, 1\}^n$ 
denotes the vector where $e^i_i = 1$ and $e^i_j = 0$ for all $j \not = i$.

\subsection{Submodular maximization on integer lattices.}
Given a bounded integer lattice $D = \{x \in \mathbb{Z}^n |l_i \leq x_i \leq u_i \ \forall i \in \{1, \ldots, n\}\}$ 
and a submodular function 
$f : D \rightarrow \mathbb{R}_+$, 
we consider the problem of maximizing $f$ on $D$:
\begin{equation}\label{SMBIL}
\max\{f(x)\mid x \in D \}.
\end{equation}
We will refer to \eqref{SMBIL} as
\textsc{Submodular maximization on a bounded integer lattice (SMBIL)}.
For ease of notation, we will from now on assume that $l_i = 0$ and 
$u_i = C \ \forall i \in \{1, \ldots, n\}$. 
Thus, we prove all results for a bounded integer lattice of the form $\{0, 1, \ldots, C\}^n$, but 
all results in this paper can be easily generalized to any bounded integer lattice as defined above. 

As mentioned before, SMBIL generalizes USM, thus the hardness of approximation for USM holds.
The integer lattice with $C = 1$ underlying USM is also called the \emph{Boolean lattice}.
Directly generalizing the Diminishing Returns property for set functions, it can be stated for the integer lattice as in \cite{SomaYoshi16}: 
$f(x + \chi_i) - f(x) \ge f(y + \chi_i) - f(y)$ for vectors $x \le y$ where $\chi_i \in {0, 1}^n$ denotes the vector whose entries are $0$ 
everywhere except for component $i$. 
They also remark that a function $f$ is DR-submodular on the integer lattice  if and only if $f$ is submodular and coordinate-wise concave, i.e.\
$f(x + \chi_i) - f(x) \ge f(x + 2\chi_i) - f(x + \chi_i)$ for any $x$ and $i$ \cite{SomaYoshiSubmodularCover}.
Our definition of submodularity differs slightly from that of Soma at al.: while we only define a submodular function on $D$, they define submodularity 
on $\Zn$, but restrict to vectors $x$ with $0\le x_i\le C$ for all optimization problems. These formulations are equivalent in the sense that 
any submodular function on $D$ can be extended to a submodular function on $\Zn$.

\subsection{Posets and Distributive Lattices}
For a partially ordered set (poset) $P = (\cN, \preceq)$, an \emph{antichain} 
is a set $S\subseteq \cN$ of incomparable elements 
and a \emph{chain} is a set $S\subseteq \cN$ where each pair of elements is comparable. 
An \emph{ideal} is a set $S\subseteq \cN$ where $s\in S$ implies 
$t \in S$ for all $t \preceq s$. By $\cD(P)$ we denote the set of all ideals for $P$. 

A \emph{lattice} is a poset $P$ any two of whose elements $x, y \in P$ 
have a \emph{meet} $x \land y$, i.e.\ a unique greatest common lower bound, and a \emph{join} $x \lor y$, i.e.\ a unique least common upper bound. 
A lattice is called distributive, if meet and join satisfy distributivity. 
One example for such a lattice is the (bounded) integer lattice with component-wise minimum and maximum as join and meet.  
An element $x$ of a lattice is called \emph{join-irreducible}, if it cannot be characterized as the join of two elements $y, z \not = x$ .

For a poset $P$,  the pair $(\cD(P), \subseteq)$ forms a distributive lattice with union and intersection as join and meet. 
Conversely, any distributive lattice is isomorphic to the lattice of ideals w.r.t.\
the induced poset on the join-irreducible elements of that lattice. This fact is known as Birkhoff's theorem \cite{Birkhoff}.
Therefore, we will assume throughout this paper that a distributive lattice is given in the form  $(\cD(P), \subseteq)$ for a poset $P$. 
For the example of the bounded integer lattice of dimension $n$, 
the corresponding lattice of ideals would be the set of ideals $\mathcal{D}(P)$ on a poset $P$ consisting of $n$ disjoint chains of length $C$. 
For a more extensive introduction to lattice theory, the reader is referred to \cite{Birkhoff}.

For a poset $P=(\cN,\preceq)$, we can consider the following generalization of SMBIL.
 \begin{equation}\label{SMDL} 
 \max\{f(S)\mid S\in \mathcal{D}(P)\},
 \end{equation}
 where
 $\mathcal{D}(P)$ is the collection of all ideals in the poset $P$ and $f$ a nonnegative submodular function on $\mathcal{D}(P)$.
 We refer to \eqref{SMDL} as \textsc{Submodular maximization on distributive lattices (SMDL)}.

 
The DR-property naturally generalizes to distributive lattices: 
$f(S + x) - f(S) \ge f(T + y) - f(T)$ for $S \subseteq T$ with $S, T \in \cD(P)$ and $x \le y$ such that $S+x, T +y \in \cD(P)$. 
This definition contains both the definition on set functions and on the integer lattice. 


 \subsection{Poset Matroids}
Integer polymatroids can be seen as a generalization of matroids to the integer lattice. 
These form special distributive lattices, as we saw above, and matroid structures can be extended further to distributive lattices as well. 
One way of doing so are poset matroids, and we are going to give a brief introduction here.
The main difference to ordinary matroids is that the ground set is a poset and matroid axioms only have to hold for ideals of this poset.
 \begin{defn}
 Given a partial order $P=(\cN,\preceq)$, a nonempty family $\cF \subseteq \cD(P)$ forms the independent sets of a poset matroid 
 if the following two properties hold:
 
 \begin{itemize}
\setlength{\itemindent}{1em}
  \item[(M1)] $Y\in \cF, X\subset Y, X \in \cD(P) \Rightarrow X\in \cF$
\item[(M2)] $X,Y \in \cF,\  |X| < |Y| \Rightarrow \exists e\in Y\setminus{X}$ with $X\cup\{e\}\in \cF$.
\end{itemize}
 \end{defn}
 
We remark that matroids are the special case of poset matroids where $P$ consists only of pairwise incomparable elements and 
integer polymatroids correspond to the special case where $P$ consists of disjoint chains. 
Therefore, by considering all subsets instead of the sets in $\cD(P)$ we obtain the usual matroid definition.

We use the usual matroid terminology for poset matroids, for example, sets in $\cF$ are called independent and 
a basis in a poset matroid is a maximal independent set.

Poset matroids were introduced by Dunstan, Ingleton and Welsh~\cite{DIW1972} and are also known as \emph{distributive supermatroids}. 
Moreover, they are a special case of \emph{ordered matroids} as introduced by Faigle \cite{Faigle1984} and several other structures.
They share many structural properties with matroids, in particular, Faigle showed that, under a consistency assumption on the cost, 
it is possible to compute a minimum cost basis using a greedy algorithm \cite{Faigle1984}. 
Moreover, Barnabei et al.\ presented a strong base exchange property \cite{BarnabeiSymmExchange} and for two poset matroids on the same poset, 
finding a  common independent set of maximum cardinality is possible in polynomial time as shown by Tardos \cite{TardosPosetIntersection}. 
She also showed that the intersection problem is NP-hard for two poset matroids on posets with different partial orders.  

However, in contrast to polymatroids, as far as we know, there exists no polyhedral description for poset matroids.

It is important to realize that poset matroids are not just the intersection of independent sets of a matroid  $(E, \cF)$ and ideals of a poset on $E$.
Consider for example the graphic matroid on the multigraph in Figure \ref{fig:posetMatroidDifference}.
Then the sets $\{e_1, e_3\}$ and $\{e_2\}$ are independent ideals,
but property $(M2)$ is not fulfilled.

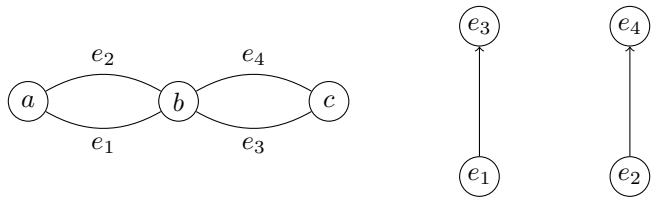
\begin{figure}[tb]
\centering
\begin{tikzpicture}[node distance=3cm]
\tikzstyle{vertex}=[circle, draw,inner sep=1pt, minimum width=15pt]
\tikzstyle{edge} = [draw]
\node[vertex] (a) at (0,0) {$a$};
\node[vertex] (b) at (2,0) {$b$};
\node[vertex] (c) at (4,0) {$c$};
\draw[edge, bend right] (a) edge node[below]{$e_1$} (b);
\draw[edge, bend left] (a) edge node[above]{$e_2$} (b);
\draw[edge, bend right] (b) edge node[below]{$e_3$} (c);
\draw[edge, bend left] (b) edge node[above]{$e_4$} (c);
\node[vertex] (e1) at (6,-1) {$e_1$};
\node[vertex] (e3) at (6,1) {$e_3$};
\node[vertex] (e2) at (8,-1) {$e_2$};
\node[vertex] (e4) at (8,1) {$e_4$};
\draw [->](e1) edge (e3);
\draw [->](e2) edge (e4);
\end{tikzpicture}
\caption{Multigraph and poset on the edges, illustrating that independent ideals in the graphic matroid do not form a poset matroid.}\label{fig:posetMatroidDifference}
\end{figure}


\section{Approximating SMBIL}\label{sec:SMBIL-Algo}
In this section, we consider unconstrained maximization of a general nonnegative submodular function on the integer lattice (SMBIL). In particular, 
we do not require DR-submodularity. 

We first present a $\nicefrac{1}{3}$-approximation for SMBIL which is inspired by \cite{DoubleGreedy}. 
Then, we show that our analysis is tight and discuss approaches for a randomized Double Greedy algorithm. 
``Double Greedy'' refers to the idea of starting with two vectors $a$ and $b$ and modifying them until both vectors are equal,
 while ensuring $f$ never decreases.
In the beginning $a_i = 0 $ and $b_i = C \ \forall i \in \{1, \ldots, n\}$, i.e.\ initially, $a$ (resp.\ $b$) is the unique minimal (maximal) element
in $D$. 
Now, we traverse the components in a fixed order, say $1, \ldots n$. For a given index $k$, we change $a_k$ and $b_k$ 
while maintaining $a \leq b$ without decreasing the submodular function $f$.
In particular, we find the best possible way to modify $a$ and $b$ while changing only component $i$. We then choose the modification which yields
the biggest increase in $f$ and set both vectors to that value. 

In the scenario where the integer lattice is bounded by 1, our Algorithm \ref{MyGreedy} coincides with the one in \cite{DoubleGreedy} 
when vectors are interpreted as characteristic vectors of sets. 

\begin{theorem}\label{ApproximationThm}
Let $f: [C]^n \rightarrow \mathbb{R}_+$ be a nonnegative submodular function. 
Then Algorithm \ref{MyGreedy} is a $\nicefrac{1}{3}$-approximation for SMBIL and has running time 
$\mathcal{O}(CnT)$ where $T$ is the time for one call to the oracle representing $f$. 
\end{theorem}

\subsection{Proof of Theorem \ref{ApproximationThm}}

The proof of Theorem \ref{ApproximationThm} relies upon two lemmas. 

\begin{algorithm2e}[htb]
\caption{Generalized Double Greedy for SMBIL}
\DontPrintSemicolon
\textbf{Input:} A bounded integer lattice defined by a bound $C$ and a dimension $n$, 
a nonnegative submodular function $f$ on 
$[C]^n$\;
\textbf{Output:} A vector $a \in [C]^n$\;
Set $a^0_j = 0, b^0_j = C \ \forall 1 \leq j \leq n $, $i = 1$\;
\For{ k = 1 \KwTo n}
{

  $\delta_{a,i} =  \max_{c \in [C]} f(a^{i-1}|a^{i-1}_k = c) - f(a^{i-1})$\;
  $\delta_{b,i} =  \max_{c \in [C]} f(b^{i-1}|b^{i-1}_k = c) - f(b^{i-1})$\;
  \If {$\delta_{a,i} \geq \delta_{b,i}$}
  {
    Let $c'$ be the maximal number among those for which $\delta_{a, i}$ is obtained.\; 
    $a^i = (a^{i-1}|a^{i-1}_k = c')$, $b^i = b^{i-1}$\;  
    $a^{i+1} = a^i$, $b^{i+1} = (b^i|b^i_k = c')$\;
  } 
  \Else
  {
    Let $c'$ be the minimal number among those for which $\delta_{b, i}$ is obtained.\; 
    $a^i = a^{i-1}$, $b^i = (b^{i-1}|b^{i-1}_k = c')$\;  
    $b^{i+1} = b^i$, $a^{i+1} = (a^i | a^i_k = c')$\;
  }
  $i = i + 2$\;  
}
\Return{ $a^{2n}$}\;
\label{MyGreedy}
\end{algorithm2e}

First, we show that Algorithm \ref{MyGreedy} really is a Greedy algorithm 
in the sense that $f$ never decreases:
\begin{lemma}\label{IncreaseLemma}
Let $f: [C]^n \rightarrow \R_+$ be a submodular function. Then for all $1\leq i \leq 2n$
and vectors $a^i$, $b^i$ as in Algorithm \ref{MyGreedy} the following holds: 
$f(a^i) \geq f(a^{i-1})$ and $f(b^i) \geq f(b^{i-1})$. 
\end{lemma}

\begin{proof}
Let $k$ be the component in which the vectors $a^{i-1}$ and $b^{i-1}$ will be changed in an iteration of the loop. 
%
%

By definition, $\delta_{a, i}$ and $\delta_{b, i}$ are nonnegative.
Suppose that $a$ is changed first, i.e.\ $\delta_{a, i} \geq \delta_{b, i}$, then  clearly, 
the lemma holds for $a^{i-1}$ and $a^i$ and $b^i = b^{i-1}$.
Now we show that the change from $b^i$ to $b^{i+1}$ does not lead to a decrease in $f$:

We have
 $f(b^{i-1}|b^{i-1}_{k} = a^i_k) + f(a^{i-1}|a^{i-1}_k = C) \geq f(a^{i-1}|a^{i-1}_k = a^i_k) + f(b^{i-1})$
 by submodularity of $f$. 
 
 As $f(a^{i-1}|a^{i-1}_k = a^i_k) \geq f(a^{i-1}|a^{i-1}_k = c)\ \forall \ a^{i-1}_k \leq c \leq C$, the above implies
 $$f(b^{i-1}|b^{i-1}_{k} = a^i_k) - f(b^{i-1}) 
 \geq f(a^{i-1}|a^{i-1}_k = a^i_k) - f(a^{i-1}|a^{i-1}_k = C) \geq 0.$$
Since $b^{i - 1} = b^i$ and $b^{i + 1} = (b^{i-1}|b^{i-1}_k = a^i_k)$, it follows that $f(b^{i+1}) \geq f(b^i)$ as desired.


The case where $\delta_{b, i} > \delta_{a, i}$ can be shown in the same way.
\end{proof}


Let $OPT$ denote a fixed optimal solution for SMBIL. 
In order to bound the value of a solution of Algorithm \ref{MyGreedy} with respect to the optimum, 
we define $OPT^i = (OPT \lor a^i) \land b^i$ analogous to \cite{DoubleGreedy}. 
Consequently, $OPT^0 = OPT$ and $OPT^{2n} = a^{2n} = b^{2n}$. 
In Lemma \ref{OptBoundLemma}, we show $f(OPT^{i - 1}) - f(OPT^i) \leq f(a^{i }) - f(a^{i-1}) +  f(b^{i} ) - f(b^{i-1})$, 
i.e.\ the decrease in $OPT^i$ (going from $OPT^O$ to $OPT^{2n}$) is bounded for each change of the vectors $a$ or $b$.

Then, as in \cite{DoubleGreedy}, Lemma \ref{OptBoundLemma} can be used to prove Theorem \ref{ApproximationThm}:

$f(OPT^0) - f(OPT^{2n}) =  \sum\limits_{i = 1}^{2n} (f(OPT^{i - 1}) - f(OPT^i ))
\leq \sum\limits_{i = 1}^{2n} \big(f(a^i) - f(a^{i - 1})\big) + \sum\limits_{i = 1}^{2n} \big(f(b^i) - f(b^{i - 1})\big)
= f(a^{2n}) + f(b^{2n}) - f(a^0) - f(b^0) \leq f(a^{2n}) + f(b^{2n}).$

This is equivalent to $f(OPT) \leq 3 f(a^{2n})$ and the algorithm returns $a^{2n}$. 
Since the running time is obvious, 
this concludes the proof of Theorem \ref{ApproximationThm} except for Lemma \ref{OptBoundLemma}: 

\begin{lemma}\label{OptBoundLemma}

For a submodular function $f: [C]^n \rightarrow \R_+$  the following holds in Algorithm \ref{MyGreedy} 
for all $1\leq i\leq 2n$ where $OPT^i:= (OPT \lor a^i) \land b^i$:

\hspace*{2 cm} $f(OPT^{i - 1}) - f(OPT^i) \leq f(a^{i }) - f(a^{i-1}) +  f(b^{i} ) - f(b^{i-1}).$ 
\end{lemma}
\begin{proof}
Let us consider the case where  $a^i \not = a^{i- 1}$ and let $k$ be the index, where the vectors differ. 
 If $a^i_k \leq OPT_k$ then $OPT^{i-1}_k = OPT^i_k$ 
 since $b^i_k = b^{i-1}_k$ and thus $OPT^{i-1} = OPT^i$.
 Since Lemma \ref{IncreaseLemma} implies that the right-hand side of the equation is nonnegative, the inequality holds. 
 
So we assume now that $a^i_k > OPT_k$.
Since $b^i_k \geq a^i_k$, we have $OPT^i_k = a^i_k$, all other entries of $OPT^{i-1}$ remain unchanged.
Moreover, $a^{i-1}_k = 0$, so $OPT^{i-1}_k = OPT_k$.
%
%
%
%
Submodularity of $f$ implies
\begin{align*}
      f(OPT^{i}) &+ f(b^{i-1}|b^{i-1}_{k} = OPT_k) \geq f(OPT^{i - 1}) + f(b^{i-1}|b^{i-1}_{k} = a^i_k) \numberthis \label{eq1}\\
\Leftrightarrow f(OPT^{i}) &+ f(b^{i-1}|b^{i-1}_{k} = OPT_k) - f(b^{i-1}) \\
      &\geq f(OPT^{i - 1}) + f(b^{i-1}|b^{i-1}_{k} = a^i_k) - f(b^{i-1}).
\end{align*}

Suppose that  $f(OPT^{i - 1}) - f(OPT^i) > f(a^{i }) - f(a^{i-1})$, otherwise we are done.

Using this assumption and \eqref{eq1} yields
\begin{align*}
&f(b^{i-1}|b^{i-1}_{k} = OPT_k) - f(b^{i-1}) + f(OPT^{i - 1}) - (f(a^{i}) - f(a^{i - 1}))\numberthis \label{eq2}\\
> &f(b^{i-1}|b^{i-1}_{k} = OPT_k) - f(b^{i-1})+ f(OPT^i)\\
\geq &f(OPT^{i - 1}) + f(b^{i-1}|b^{i-1}_{k} = a^i_k) - f(b^{i-1}).
\end{align*} 

But we have $f(a^{i}) - f(a^{i-1 }) \geq f(b^{i-1}|b^{i-1}_{k} = c) - f(b^{i-1}) \ \forall \ a^{i-1}_k \leq c \leq b^{i-1}_k$.
This is true by design of the algorithm for the first change in an iteration of the loop. 
For the second change in an iteration, the other vector (in this case $b^{i - 1}$) cannot improve further, so the claim holds as well. 

Setting $c = OPT_k$ and \eqref{eq2} imply
\begin{align*}
 0 &\geq f(b^{i-1}|b^{i-1}_{k} = OPT_k) - f(b^{i-1}) - (f(a^{i}) - f(a^{i - 1})) \numberthis \label{eq2}\\ 
 &> f(b^{i-1}|b^{i-1}_{k} = a^i_k) - f(b^{i-1}). 
\end{align*} 

But in the proof of Lemma \ref{IncreaseLemma}, we have shown that 
$f(b^{i-1}|b^{i-1}_{k} = a^i_k) - f(b^{i-1}) \geq 0$, so \eqref{eq2} is a contradiction.
 
The case $b^i \not = b^{i- 1}$ can be treated analogously and if neither case applies, then clearly 
$OPT^i = OPT^{i-1}$ as well.  
\end{proof}


\subsection{The guarantee of $\frac{1}{3}$ is tight}\label{subsec:TightExampleSec}
Since Algorithm \ref{MyGreedy} generalizes the deterministic algorithm in \cite{DoubleGreedy}, 
the tight example they provide also works for our algorithm. But our example has a few additional properties: 
First, even if we do not prescribe the order of the components in advance 
and instead choose the index that results in the biggest increase in $f$ for each step, 
the algorithm does not yield a better solution (see Theorem \ref{AlgoTightThm} below) which is not true for 
 the example provided in \cite{DoubleGreedy}.
However, it is not possible to attain a $\frac{1}{2}$-approximation by choosing a particular good order. 
This was shown by Huang and Borodin in \cite{BorodinMyopicBounds}. 
They examine classes of Double Greedy algorithms and provide an upper bound of 
$\frac{2}{3\sqrt{3}} + \epsilon \approx 0.385+\epsilon$ on approximation ratios achievable by those algorithms. 
These include in particular the algorithm described above. 
Moreover, we will show in the next subsection that our example also has implications when trying to generalize the 
$\nicefrac{1}{2}$-approximation presented in \cite{DoubleGreedy}.
%
\begin{theorem}
\label{AlgoTightThm}
 For an arbitrarily small constant $\varepsilon > 0$ there exists a submodular function 
 for which Algorithm \ref{MyGreedy} provides only a $\frac{1}{3} + \varepsilon$-approximation.
 This result still holds if the components are ordered such that at any time the
 chosen component yields the biggest increase in $f$.
\end{theorem}
\begin{figure}[htb]

	 \tikzstyle{vertex}=[circle, draw,
                        inner sep=1pt, minimum width=13pt]
	\tikzstyle{edge} = [draw,thick]
	\centering
\tikzstyle{vertex}=[circle, draw,
                        inner sep=1pt, minimum width=13pt]
	\tikzstyle{edge} = [draw,thick]
	
  	\centering
	\begin{tikzpicture}[scale = 0.8]
	\foreach \x/\y/\nr/\f/\set in {{0/0/0/0/\tiny $(0, 0, 0)$}, {-2/1/1/1/\tiny $(1, 0, 0)$}, {0/1/2/1/\tiny$(0, 1, 0)$}, {2/1/3/1/\tiny$(0, 0, 1)$}, 
				  {-5/3/4/\tiny $1+\varepsilon$/\tiny $(2, 0, 0)$}, {-3/3/5/2/\tiny $(1, 1, 0)$},
				  {-1/3/6/2/\tiny $(1, 0, 1)$}, {1/3/7/2/\tiny $(0, 1, 1)$}, 
				  {3/3/8/\tiny $1+\varepsilon$/\tiny $(0, 2, 0)$}, {5/3/9/\tiny $1+\varepsilon$/\tiny $(0, 0, 2)$}, 
				  {-6/5/10/2/\tiny $(2, 1, 0)$}, {-4/5/11/\tiny $1+\varepsilon$/\tiny $(2, 0, 1)$}, 
				  {-2/5/12/3/\tiny $(1, 1, 1)$},{0/5/13/2/\tiny $(1, 2, 0)$}, 
				  {2/5/14/\tiny $1+\varepsilon$/\tiny $(0, 2, 1)$},{4/5/15/2/\tiny $(1, 0, 2)$}, {6/5/16/2/\tiny $(0, 1, 2)$}, 
				  {-5/7/17/2/\tiny $(2, 1, 1)$}, {-3/7/18/\tiny $1+\varepsilon$/\tiny $(2, 2, 0)$}, 
				  {-1/7/19/\tiny $1+\varepsilon$/\tiny $(2, 0, 2)$}, 
				  {1/7/20/\tiny $1+\varepsilon$/\tiny $(0, 2, 2)$}, {3/7/21/2/\tiny $(1, 2, 1)$}, 
				  {5/7/22/2/\tiny $(1, 1, 2)$}, {-2/9/23/1/\tiny $(2, 2, 1)$}, {0/9/24/1/\tiny $(2, 1, 2)$}, 
				  {2/9/25/1/\tiny $(1, 2, 2)$}, {0/10/26/0/\tiny $(2, 2, 2)$}}
	  {
	   \node[vertex](v\nr) at (\x, \y){\scriptsize \f};
	   \draw(\x + 0.25, \y) node[anchor = west]{\scriptsize \set};
	  }  
	 \foreach \i/\j in {{0/1}, {0/2}, {0/3}, {1/4}, {1/5}, {1/6}, {2/5}, {2/7}, {2/8}, {3/6}, {3/7}, {3/9}, 
			     {4/10}, {4/11},{5/10}, {5/12}, {5/13}, {6/11}, {6/12}, {6/15}, {7/12}, 
			     {7/14}, {7/16}, {8/13}, {8/14}, {9/15}, {9/16}, 
			     {10/17}, {10/18}, {11/17}, {11/19}, {12/17}, {12/21}, {12/22}, {13/18}, {13/21}, 
			     {14/20}, {14/21}, {15/19}, {15/22}, {16/20}, {16/22}, 
			     {17/23}, {17/24}, {18/23}, {19/24}, {20/25}, {21/23}, {21/25}, {22/24}, {22/25}, 
			     {23/26}, {24/26}, {25/26}}
 	      \draw (v\i) edge (v\j);

	\end{tikzpicture}

\caption{Each node corresponds to a vector in $\{0, 1, 2\}^3$, written next to it, with the value of $f$ within the node. 
The picture can be read similar to Hasse diagrams of posets: 
There is a line connecting two vectors, if their $L_1$-distance equals 1. If we think of the edges as being directed upwards, 
then the meet of two vectors is their biggest common predecessor, and their join the smallest common successor.}
\label{fig:TightExample}
\end{figure}
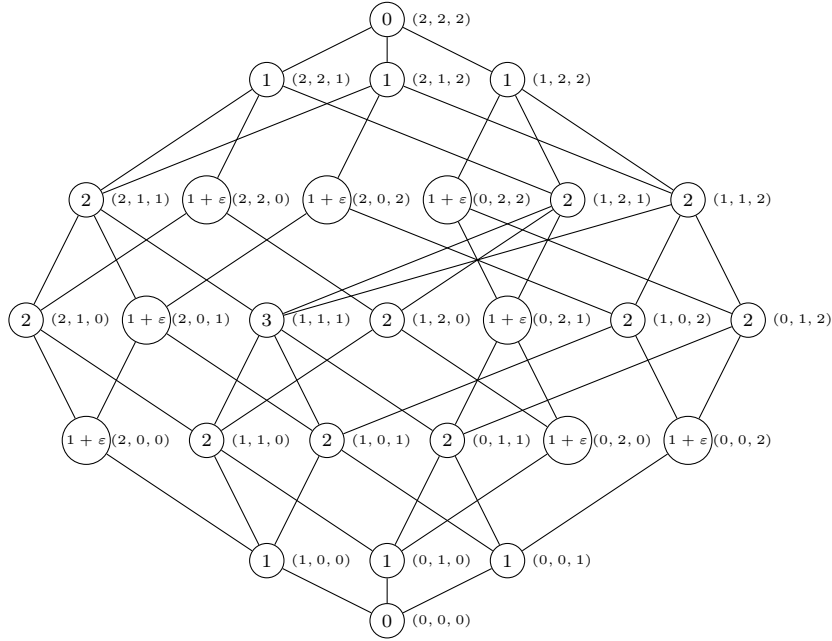

\begin{proof}
Consider the following submodular function $f: \{0, 1, 2\}^3 \rightarrow \R_+ $. 
Let $f'(a) = \min \{|\{ a_i: a_i > 0\}|, |\{ a_i: a_i < 2\}| \}$. Now we define 
$f(a) = 1+\varepsilon$ if $a$ consists of either the entries $2, 0, 2$ or $2, 0, 0$ in any order. 
We set $f(2, 0, 1) = f(0, 2, 1) = 1+\varepsilon$. For all other vectors, we set $f(a) = f'(a)$. 
It can be checked that $f$ is indeed submodular. 
To give a better intuition, Figure \ref{fig:TightExample} illustrates this function.

We analyze Algorithm \ref{MyGreedy} for this instance: 
Obviously, the optimum is the vector $(1, 1, 1)$ of value 3, 
but Algorithm \ref{MyGreedy} terminates with a set of value $1 + \varepsilon$: 

In the first iteration, the maximal possible gain in $f$ is $1+\varepsilon = \delta_{a, 1} = \delta_{b, 1}$. 
So we we set $a^1_{1} = 2$ and $b^2_{1} = 2$. 
Next, $\delta_{b, 3} = 1+\varepsilon > \delta_{a, 3} = 0$ and thus $b^3_{2} = 0$ and $a^4_{2} = 0$. 

Therefore, we now have $a^4 = (2, 0, 0)$ and $b^4 = (2, 2, 0)$, both of value $1 + \varepsilon$ 
and the algorithm returns the vector $a^6 = b^6 = (2, 2, 0)$ of value $1 + \varepsilon$. 

Since the maximal possible gain in $f$ is $1+\varepsilon$ for the first two indices independent of the order,
the same order could have been chosen by an algorithm that always processes the index which maximizes the gain in $f$.
Moreover, the tie-braking rule for $\delta$ 
does not influence the value of the output here.
\end{proof}

While Algorithm \ref{MyGreedy} went ``too far'' when choosing the next entry, 
a variant of the algorithm that only changes the vectors by one at a time can be arbitrarily bad here. To see this, consider the example on $[2]^2$ in Figure 
\ref{fig:BadLatticeOneGreedy}. If our algorithm was only allowed to add $1$ to $a$ or delete $1$ from $b$ and chose the better of these two options, 
the output could not be $(0, 2)$ which has an arbitrarily high value of $x$. 
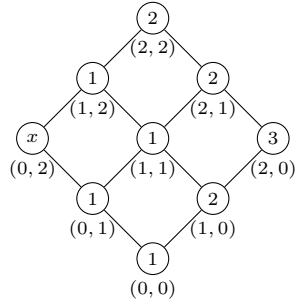
\begin{figure}[tb]

	 \tikzstyle{vertex}=[circle, draw,
                        inner sep=1pt, minimum width=12pt]
	\tikzstyle{edge} = [draw,thick]
	\centering
	\begin{tikzpicture}[scale = 0.8]
	\foreach \x/\y/\nr/\f in {{0/0/0/1}, {-1/1/1/1}, {1/1/2/2}, {-2/2/3/$x$}, {0/2/4/1}, {2/2/5/3}, 
				  {-1/3/7/1}, {1/3/8/2}, {0/4/11/2}}
	   \node[vertex](v\nr) at (\x, \y){\scriptsize\f};

	   \draw (0,-0.2) node[anchor = north]{\scriptsize $(0,0)$};
	   \draw (-1,0.8) node[anchor = north]{\scriptsize $(0, 1)$};
	   \draw (1,0.8) node[anchor = north]{\scriptsize $(1, 0)$};
	   \draw (-2, 1.8) node[anchor = north]{\scriptsize $(0, 2)$};
	   \draw (0, 1.8) node[anchor = north]{\scriptsize $(1, 1)$};
	   \draw (2, 1.8) node[anchor = north]{\scriptsize $(2, 0)$};
	   \draw (-1,2.8) node[anchor = north]{\scriptsize $(1, 2)$};
	   \draw (1,2.8) node[anchor = north]{\scriptsize $(2, 1)$};
	   \draw (0,3.8) node[anchor = north]{\scriptsize $(2, 2)$};
	 \foreach \i/\j in {{0/1}, {0/2}, {1/3}, {1/4}, {2/4}, {2/5}, {3/7}, {4/7}, {4/8}, {5/8},
			      {7/11},{8/11}}
	      \draw (v\i) edge (v\j);
	  \draw [white] (3,1) circle (1mm);    
	   
	\end{tikzpicture}
\caption{The values in the nodes give a submodular function for any $x \geq 1$.}
\label{fig:BadLatticeOneGreedy}
\end{figure}

\subsection{Difficulties in Designing Algorithms}\label{sec:BadRandom}
We already mentioned that submodularity is a much weaker property on the integer lattice than on set functions. 
This also makes it hard to generalize some algorithms without assuming additional properties like diminishing returns. 
Consider for example the one-dimensional integer lattice, on which any function is submodular. Therefore, we cannot expect 
running times that are better than pseudopolynomial for any algorithm with bounded approximation ratio, unless the function fulfills other properties as well.
This example also illustrates that some structural properties that are key to designing good algorithms on set functions do not hold any more. 
One such result is used in the local search approach presented in \cite{LocalSearch}. The authors rely on a generalization of a result about local optima, 
stating that the value of a local optimum $x$ is at least as large as the value of all sub- and supersets. 
For the integer lattice, this statement is no longer true. 
Moreover, ``continuous Greedy algorithms'' that were used in various settings and also provided a better approximation ratio for USM,
rely on component-wise concavity along non-negative directions 
for the so-called multilinear relaxation. Since such a relaxation coincides with the function values on integer points, this property cannot be obtained 
without stronger assumptions like DR. 

As a generalization of the Double Greedy approach was possible, 
it seems reasonable to ask whether further results using similar techniques can also be generalized to SMBIL. 
For USM, \cite{DoubleGreedy} also presents a randomized ``Double Greedy'' algorithm which gives a guarantee of $\frac{1}{2}$. 
They decide with probability proportional to the increase in $f$ whether to add a given element to one set or delete it from the other. 
In the context of incidence vectors, this is equivalent to choosing whether entries $a_i$ and $b_i$ are set to $0$ or $1$. 
We show that a similar analysis cannot work if we adapt their idea to our algorithm.
We consider two natural strategies of generalizing the randomized algorithm above. One is the following:
For given vectors $a, b$ and an index $i$ we can consider all possibilities to increase $a_i$ or decrease $b_i$ such that
$a \leq b$ remains true and choose one of these possibilities at random (again proportional to the increase in $f$). 
We will show this leads to arbitrarily bad solutions. 

The other alternative is more similar to our previous algorithm: We determine the best choice for $a$ and $b$ and then choose
between the two options with probability proportional to the $\delta$-values, i.e.\ the maximal possible gain in $f$.
While it is possible that this actually is a $\nicefrac{1}{2}$-approximation, 
we can so far only show that this randomized algorithm gives the same guarantee as the deterministic version.
Indeed, we will show that an analysis as in \cite{DoubleGreedy} which bounds the expected decrease of 
$f(OPT^i)$ in each step cannot prove a guarantee better than $\nicefrac{1}{3}$ by examining the example 
presented in Figure \ref{fig:TightExample}. 

So, while it might still be true that this randomized version of our algorithm actually is a $\nicefrac{1}{2}$-approximation 
(for example, the worst case expected value in the given example is $1.5 + \varepsilon$ while the optimum is $3$), 
we would require an analysis that takes a more global view. 

We should note that for the Boolean lattice, both these approaches are identical and  correspond to the randomized algorithm in 
\cite{DoubleGreedy}. We now analyze both algorithms. 

\subsubsection{Randomized choice over all possible solutions}\label{subsec:RandBadAllPossible}

\begin{lemma}\label{AllRandChoicesBadLemma}
Consider the following randomized version of Algorithm 1: Assume an order of indices and 
consider index $i$. While $a_i < b_i$, randomly set either $a = (a|a_i = c)$ or $b = (b|b_i = c)$ where the choice is made over all 
$a_i \leq c\leq b_i$ proportional to the increase in $f$ if it is positive. 
If the increase in $f$ is non-positive for all possibilities, we arbitrarily choose an option where the increase in $f$ is $0$, 
but $a$ or $b$ are changed.     
This algorithm for SMBIL can perform arbitrarily bad.

\end{lemma}

\begin{proof}
First, we remark that in the case where the increase in $f$ is non-positive everywhere, a choice as above exists.    
We now analyze the following instance: Given  $C \in \mathbb{N}$, $\varepsilon > 0$, 
 we define a submodular function $f: [C]^2 \rightarrow \mathbb{R}$ as follows

$f(x) = \begin{cases}
	0,  & \text{if } x = (0, 0) \text{ or } x = (C, C) \\
	1 & \text{if }	x = (C, 0)\\
	\varepsilon & \text{else}
	\end{cases}$.
	
Our randomized algorithm starts with $a = (0, 0)$ and $b = (C, C)$. 
We start with index 1. Until $a_1 = b_1$, we change entry $a_1$ or $b_1$ and choose a value such that $a_1 \leq b_1$ 
at random proportional to the increase in $f$. 
I.e., for the first step, we have $2C$ options where the increase in $f$ is positive, 
all but one (setting $a_1 = C$) will be taken with probability $\frac{\varepsilon}{(2C-1)\varepsilon + 1}$. 


After the first step, either $b_1 < C$ or $a_1 > 0$. In the second case, if $a_1 < C$, the only options now to increase $f$ are 
changing $b_1$ or setting $a_1 = C$.  
Therefore, after two steps either $a_1 = C$ or $b_1 < C$ and depending on which of these holds, the return value will be $1$ or $\varepsilon$. 
Thus, the probability for the algorithm to return a vector of value $\varepsilon$ is
\begin{align*}
&Pr[b \text{ is changed first}] + \sum\limits_{i = 1}^{C-1}Pr[ a_1 \text{ is set to }i \text{ first, then } b \text{ is changed}]\\
  = &\frac{C\varepsilon}{(2C-1)\varepsilon + 1} + 
\frac{\varepsilon}{(2C-1)\varepsilon + 1} \sum\limits_{i = 1}^{C-1} \frac{(C-i)\varepsilon}{(C-i)\varepsilon + 1-\varepsilon}
\end{align*}
We now show that this expression converges to one for $C\rightarrow \infty$, thus showing that 
the expected return value of the algorithm converges to $\varepsilon$.

To achieve that, it suffices to find a lower bound converging to one. After rewriting the above expression we obtain: 
 \begin{align*}
&\frac{C\varepsilon}{(2C-1)\varepsilon + 1} + 
    \frac{\varepsilon}{(2C-1)\varepsilon + 1} \sum\limits_{i = 1}^{C-1} \frac{(C-i)\varepsilon}{(C-i)\varepsilon + 1-\varepsilon}\\
\geq &\frac{C\varepsilon}{2C\varepsilon + 1} 
    + \frac{\varepsilon}{(2C-1)\varepsilon + 1} \sum\limits_{i = 1}^{C-1}\frac{(C-i)\varepsilon}{(C-i)\varepsilon + 1}\\
= &\frac{C\varepsilon}{2C\varepsilon + 1}
    + \frac{\varepsilon}{(2C-1)\varepsilon + 1} \sum\limits_{i = 1}^{C-1} \Big(1 - \frac{1}{(C-i)\varepsilon + 1}\Big)\\
\geq &\frac{C\varepsilon}{2C\varepsilon + 1}
    + \frac{(C-1)\varepsilon}{(2C-1)\varepsilon + 1} 
    - \frac{\varepsilon}{(2C-1)\varepsilon + 1} \sum\limits_{i = 1}^{C-1} \frac{1}{(C-i)\varepsilon}\\
=  &\frac{C\varepsilon}{2C\varepsilon + 1}
    + \frac{(C-1)\varepsilon}{(2C-1)\varepsilon + 1} 
    - \frac{1}{(2C-1)\varepsilon + 1} \sum\limits_{i = 1}^{C-1} \frac{1}{i}\\
\geq & \frac{C\varepsilon}{2C\varepsilon + 1}
     + \frac{(C-1)\varepsilon}{(2C-1)\varepsilon + 1} 
     - \frac{\ln(C-1) + 1}{(2C-1)\varepsilon + 1} 
     \rightarrow \frac{1}{2} + \frac{1}{2} - 0 = 1 \text{ for }C \rightarrow \infty\\
\end{align*}
For the last inequality, we used the fact that the partial sums of the harmonic series can be bounded by $\ln$, 
the rest is simple arithmetic.  
\end{proof}
Note that variations like randomizing over combinations of values for $a_i$ and $b_i$ 
and choosing proportional to the sum of the increases in $f$ show a similar behavior. 

\subsubsection{Randomized choice over the best possible solutions}\label{subsec:RandBadBestOptions}

\begin{lemma}
Consider the following randomized version of Algorithm 1: Instead of deciding to change $a^{i-1}$ or $b^{i-1}$ depending on whether
$\delta_{a, i} \geq \delta_{b, i}$, 
randomly change $a$ or $b$ proportional to $\max\{0, \delta_{a, i}\}$ and $\max\{0, \delta_{b, i}\}$.
Then adapt the other vector as before. 
For this algorithm, there is no constant $c < 1$ such that 
$E[f(OPT^{i-1}) - f(OPT^i)] \leq c\cdot E[f(a^i) - f(a^{i-1}) + f(b^i) - f(b_{i-1})]$.
\end{lemma}

Unlike in \cite{DoubleGreedy}, where the above statement is true with $c = \nicefrac{1}{2}$,
an analysis that bounds the expected decrease of $OPT^i$ by the expected increase in $a$ and $b$ in each step cannot
yield an approximation factor better than $\nicefrac{1}{3}$. 

\begin{proof}
As before, $OPT^i$ is defined as $(OPT \lor a^i)\land b^i$ for a fixed optimal solution $OPT$.  
Consider the example presented in the previous section in Figure \ref{fig:TightExample}. 
No matter how we choose the first index $j$, $a_j = 2$ or $b_j = 0$ both with equal probability. 
Thus, $OPT^1$ consists of the entries $0, 1, 1$ or $2, 1, 1$ and thus has value 2 in both scenarios which implies 

$E[f(OPT^0) - f(OPT^1)] = \frac{1}{2}\cdot((3 - 2) + (3-2)) = 1$ and    
$E[f(a^1) - f(a^0) + f(b^1) - f(b_0)] = 2\cdot \frac{1}{2} \cdot (1 + \varepsilon)$.

Note that the statement is true no matter how the order of indices is chosen.
\end{proof}

\section{DR-submodularity on Distributive Lattices: Algorithms}\label{sec:SM-DR-DL-algos}
In the rest of this paper, we consider more general domains, namely distributive lattices, but only DR-submodular functions. 
In this section, we examine several problems, that are well-understood for set functions and DR-submodular functions on the integer lattice 
and propose algorithms.
First, we investigate maximizing monotone functions subject to a poset matroid constraint. 
Moreover, we consider the special case of uniform poset matroids, 
i.e.\ a cardinality constraint. Then, we generalize the double greedy approach to the unconstrained case.    

\subsection{Monotone maximization subject to a poset matroid constraint}

Given a poset $P = (\cN, \preceq)$,  let $f: \cD(P) \rightarrow \R_+$ be DR-submodular and monotone.
For a poset matroid $(P, \cF)$, we consider the problem
$$\max\{f(S) : S \in \cF\}.$$ 
We show a $\nicefrac{1}{2}$-approximation guarantee for a general poset matroid and 
a guarantee of $(1 - \nicefrac{1}{e})$ for a cardinality constraint, which is the special case of a uniform poset matroid. 
These guarantees were shown by Fisher, Nemhauser and Wolsey for matroid constraints on the Boolean lattice (\cite{FisherGreedyI},\cite{FisherGreedyII}) and
 the cardinality result is also presented in \cite{KrauseSurvey}.  
For cardinality constraints, this is tight in the sense that any algorithm evaluating $f$ 
a polynomial number of times cannot yield a better approximation guarantee than $(1 - \nicefrac{1}{e})$ \cite{FisherGreedyI}. 
There is also a $(1 - \nicefrac{1}{e})$-approximation for a matroid constraint (\cite{MonotoneMatroid1}) on the Boolean lattice 
and a $(1 - \nicefrac{1}{e} - \epsilon)$-approximation for a polymatroid constraint on the integer lattice.
However, these are based on continuous greedy algorithms and rounding a vector in a downward closed polytope, 
and it is not clear how to generalize these techniques for distributive lattices as we have no appropriate polyhedral description of poset matroids.  

We start by describing and analyzing the algorithm for poset matroid constraints. 
Then, we give a better analysis for a cardinality constraint. 
The algorithm is a simple Greedy procedure: Starting from the empty set, 
we consider a minimal element among those not yet processed in the partial order
and choose the one that yields the biggest increase in $f$, compared to the current set. 
If including this element would result in an independent set w.r.t.\ the poset matroid, 
we add it, otherwise, we discard the element. This procedure yields a linear extension $x_1, \ldots x_n$ of the partial order. 
A formal description can be found in Algorithm \ref{Algo:SMMatroid}. 

\begin{algorithm2e}
\DontPrintSemicolon
\textbf{Input:} A poset matroid $(P, \cF)$, $f: \cD(P) \rightarrow \R_+$ DR-submodular and monotone\;
\textbf{Output:} A set $S \in \cF$\;

 Set $S = \emptyset$\; 
 \For {t = 1 \KwTo n}{
    Select $x_t = \argmax_{x \in \cN, S+x \in \cD(P)}{f(S+x) - f(S)}$ \;
    \If{$S + x_t \in \cF$}
    {
      $S = S+x_t$\;
    }
    \Return S\;
 }
 \caption{Greedy with poset matroid constraint.}\label{Algo:SMMatroid}
\end{algorithm2e}

\begin{theorem}\label{thm:monotoneDRDLApprox}
Algorithm \ref{Algo:SMMatroid} is a $\nicefrac{1}{2}$-approximation for monotone DR-submodular lattice maximization subject to a matroid constraint
and a $(1 - \nicefrac{1}{e})$-approximation for a cardinality constraint.    
\end{theorem}

For the proof, we will use a known lemma, which we restate here: 

\begin{lemma}[\cite{FisherGreedyII}]\label{lemma:FisherArithmetic}
 Let nonnegative numbers $\sigma_i, \rho_i$ for $i \in 1, \ldots, K$ with $\rho_{i - 1} \ge \rho_i$ for all $i$ and 
 $\sum_{i = 1}^{t}\sigma_i \le t$ for $t \in 1, \ldots, K$ be given. 
 Then $\sum_{i = 1}^{K}\sigma_i\rho_i \le \sum_{i = 1}^{K}\rho_i$
\end{lemma}

\begin{proof}[Proof of Theorem \ref{thm:monotoneDRDLApprox}]
Let $K$ be the size of the set $S$ returned by our Greedy algorithm.  
For ease of notation, we introduce $S^0 \subsetneq S^1 \subsetneq \ldots \subsetneq S^K$ with $S^0 = \emptyset$ and $S^K = S$. 
That means, $S^i$ marks the progression of $S$ throughout the algorithm. Let $OPT$ denote a fixed optimal solution. Now we go through
the total ordering $x_1, \ldots, x_n$ and define $\sigma_{i - 1}$ as the number of elements in $OPT$ that we considered before reaching $S_i$, 
starting from the last element added to $S^{i - 1}$. 

More formally, let $r^{i}$ be the minimal index such that $S^{i} \subset \{x_1, \ldots x_{r^{i}}\}$, for convenience we define $r_{K + 1} = n + 1$.  
In other words, $x_1, \ldots, x_r^i$ are the elements considered by the Greedy algorithm at the point where $S$ is set to $S^i$ 
and in particular $\{x_{r^i}\} = S^i \setminus S^{i-1}$. 
Then $\sigma_{i} = |OPT\cap \{x_{r^{i}} , \ldots, x_{r^{i + 1} - 1}\}|$. 
We illustrate this definition in an example in Figure \ref{fig:DRMatroidExample}. 

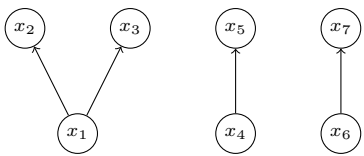
\begin{figure}[tb]
\centering
\begin{tikzpicture}[scale = 0.7]
\tikzstyle{vertex}=[circle, draw,inner sep=1pt, minimum width=15pt]
\tikzstyle{edge} = [draw]
\node[vertex] (x1) at (0,0) {\scriptsize $x_1$};
\node[vertex] (x4) at (3,0) {\scriptsize $x_4$};
\node[vertex] (x6) at (5,0) {\scriptsize $x_6$};
\node[vertex] (x2) at (-1,2) {\scriptsize $x_2$};
\node[vertex] (x3) at (1,2) {\scriptsize $x_3$};
\node[vertex] (x5) at (3,2) {\scriptsize $x_5$};
\node[vertex] (x7) at (5,2) {\scriptsize $x_7$};

\draw [->](x1) edge (x2);
\draw [->](x1) edge (x3);
\draw [->](x4) edge (x5);
\draw [->](x6) edge (x7);
  
 \end{tikzpicture}
\caption{Partial order with induced total order given by indices. Consider $OPT = \{x_4, x_5, x_6\}$ and $S = \{x_1, x_3, x_6\}$, then
$\sigma_1 = 0, \sigma_2 = 2, \sigma_3 = 1$ and $r^2 = 3$.}\label{fig:DRMatroidExample}
 
\end{figure}

We start by showing that $\sigma$ fulfills the conditions of the above Lemma, that is $\sum_{i = 1}^{t}\sigma_i \le t$ for $t \in 1, \ldots, K$. 
By construction of our Greedy algorithm, $S^t$ is a maximal independent set in $\{x_1, \ldots, x_{r^{t + 1} - 1}\}$. 
Therefore, by poset matroid properties, it also has maximum size and any independent subset of $\{x_1, \ldots, x_{r^{t + 1} - 1}\}$ 
contains at most $|S^t| = t$ elements.
As the numbers $\sigma_1, \ldots, \sigma_K$ are just one way of counting the elements of $OPT$ by grouping them into intervals on the given linear extension 
(w.l.o.g.\ $r^1 = 1$), we obtain
$$\sum_{i = 1}^t \sigma_{i} = |OPT \cap \{x_1, \ldots, x_{r^{t + 1} - 1}\}| \le t.$$ 
The last inequality holds as OPT is an independent set and thus  
$OPT \cap \{x_1, \ldots, x_{r^{t + 1}- 1}\}$ is independent as well. 

Next, we show that the differences $\rho_i:= f(S^i) - f(S^{i - 1})$ decrease as $i$ increases. 
We distinguish two cases. 
If $x_{r^{i}}$ and $x_{r^{i + 1}}$ are not comparable, we could have added $x_{r^{i + 1}}$ to $S^{i - 1}$, 
and thus, $x_r^i$ must have been the better option at that moment. 
Note that $S^{i - 1} + x_{r^{i + 1}}$ must be independent by poset matroid property $(M1)$, 
as $S^{i - 1} +x_{r^i} + x_{r^{i + 1}}$ is independent. 
Therefore,  
\begin{align*}
\rho_{i-1} &= f(S^{i - 1} + x_{r_i})  - f(S^{i - 1}) \ge  f(S^{i - 1} + x_{r^{i + 1}}) - f(S^{i - 1}) \\
	  &\ge f(S^{i - 1} +x_{r^i} + x_{r^{i + 1}}) - f(S^{i - 1} + x_{r^i}) = \rho_i,
\end{align*}
 where the second inequality follows from DR-submodularity.  
 
Otherwise,  $x_{r^{i}} \preceq x_{r^{i + 1}}$ and $\rho_{i -1} \ge \rho_{i}$ follows directly from DR-submodularity. 
Moreover, by monotonicity of $f$, $\rho_i \ge 0$ for all $i$. 
Therefore, we can apply Lemma \ref{lemma:FisherArithmetic} for $\rho_i$ and $\sigma_i$. 

Now, we consider only the elements of $OPT$ that are not chosen by our algorithm, 
ordered as in the linear extension given by $x_1, \ldots, x_n$. 
Let $\{a_1, \ldots, a_\ell\} = OPT \setminus S$, ordered the same way as the $x_i$. 
We will partition this set into the subsets contributing to the different $\sigma_t$
and to simplify notation we define these subsets as
$OPT(S, t) :=  OPT\setminus S \cap \{x_{r^{t}} , \ldots, x_{r^{t + 1} - 1}\}$. 

Now, we first use monotonicity of $f$ and then rewrite.

\begin{align*}
f(OPT) &\le f(OPT \cup S) = f(S) + \sum_{i = 1}^\ell (f(S \cup \{a_1, \ldots a_i\}) \! -\! f(S \cup \{a_1, \ldots a_{i-1}\}))\\
       &=  f(S) +\sum_{t = 1}^K \sum_{a_i \in OPT(S, t)}
		      (f(S \cup \{a_1, \ldots a_i\}) \! -\! f(S \cup \{a_1, \ldots a_{i-1}\}))\\
       &\overset{(*)}{\le} f(S) +\sum_{t = 1}^K \sum_{a_i \in OPT(S, t)} 
			(f(S^{t - 1} + x_{r^t})  - f(S^{t-1}))\\ 
       &=     f(S) +\sum_{t = 1}^K \sigma_i(f(S^{t - 1} + x_{r^t})  - f(S^{t-1}))\\ 
       &\overset{(**)}{\le} f(S) + \sum_{t = 1}^K (f(S^{t - 1} + x_{r^t})  - f(S^{t-1}) = 2f(S) - f(\emptyset) \le 2f(S)
\end{align*}
Inequality $(*)$ is a consequence of DR-submodularity and the Greedy property: 
Consider some $a_i \in OPT(S, t)$. 
If $x_{r^t} \preceq a_i$, the inequality we use follows directly from DR-submodularity
as $S^{t-1} \subset S$. 
Otherwise, $a_i$ and $x_{r^t}$ are incomparable and there is a minimal element $a \in \{x_{r^t}, \ldots, x_n\}$ with  $a \prec a_i$. 
By DR-submodularity, $f(S \cup \{a_1, \ldots, a_i\})  - f(S \cup \{a_1, \ldots a_{i-1}\}) \le f(S^{t - 1} + a) - f(S^{t - 1})$. 
But adding $x_{r(t)}$ maximized the increase in $f$ and thus, this is at most  $f(S^{t - 1} + x_{r^t}) - f(S^{t - 1})$ 
(as $a$ is not considered earlier in the algorithm). 

Inequality $(**)$ is a direct consequence of Lemma \ref{lemma:FisherArithmetic} with $\sigma$ and $\rho$ as defined above, 
and we already showed that the requirements for applying the lemma are fulfilled. 

This proves that our Algorithm \ref{Algo:SMMatroid} is a $\nicefrac{1}{2}$-approximation. 

If the poset matroid is the uniform matroid, the constraint is a cardinality constraint, i.e.\ $|S| \le k$. 
For that case,  we can use a similar argument to obtain a better bound: 
Consider some $S_i$ and define $a_1, \ldots a_\ell $ as before with $S_i$ instead of $S$, i.e.\ $\{a_1, \ldots, a_{\ell}\} = OPT \setminus S_i$. 
Clearly, $\ell \le k$. 

Analogous to the previous inequalities, we can obtain 
\begin{align*}
f(OPT) &\le f(OPT\cup S_i) = f(S_i) \!+\! \sum_{j = 1}^\ell (f(S \cup \{a_1, \ldots, a_j\}) \! -\! f(S \cup \{a_1, \ldots a_{j-1}\}))\\
&\le f(S_i) + \sum_{j = 1}^\ell (f(S_i + x_{r^i})  - f(S_i)) \le f(S_i) + k(f(S_i + x_{r^i})  - f(S_i))
 \end{align*}
As shown in \cite{KrauseSurvey}, this inequality directly implies that our Greedy algorithm 
is a $(1 - \nicefrac{1}{e})$-approximation. 
For the sake of completeness, we include the proof here as well. 

Define $\delta_i := f(OPT) - f(S_i)$ and then rewrite the above inequality:
$\delta_i \le k(\delta_i - \delta_{i + 1})$. We rearrange to obtain
$\delta_{i + 1} \le \big(1 - \frac{1}{k}\big)\delta_i$. 
 It follows inductively that $\delta_{j} \le ( 1- \frac{1}{k})^j\delta_0$ for all $j$. Moreover, 
 non-negativity of $f$ implies  $\delta_0 = f(OPT) - f (\emptyset) \le f(OPT)$. Now we use the fact that
 $1 - x \le e^{-x}$ for all $x \in \R$, resulting in
 $$\delta_j \le (1 - \frac{1}{k}\big)^j\delta_0 \le e^{-\frac{j}{k}}f(OPT).$$
 Resubstituting  $\delta_j$ then yields $f(S_j) \ge (1 - e^{-\frac{j}{k}})f(OPT)$ which implies the approximation guarantee. 
\end{proof}

 The same technique as in the proof of theorem \ref{thm:monotoneDRDLApprox} can also be used to generalize this result to a collection of $s$ poset matroid
 constraints to obtain an approximation ratio of $\frac{1}{s + 1}$.

\subsection{Unconstrained Maximization}
In this subsection, we again consider the randomized Double Greedy algorithm by Buchbinder et al.\ \cite{DoubleGreedy}. 
We will show that DR-submodularity seems to be a crucial property to obtain good approximation guarantees using that algorithm. 
In contrast to SMBIL, all proofs from \cite{DoubleGreedy} can easily be generalized to obtain a $\nicefrac{1}{2}$-approximation
for DR-SMDL. 
Note that this only results in a pseudopolynomial algorithm for the integer lattice. 
However, Soma and Yoshida showed how to adapt the Double Greedy to obtain a polynomial randomized $\nicefrac{1}{2}$-approximation
for DR-SMBIL \cite{SomaDR-SMBIL}. 
\begin{theorem}\label{thm:AlgoDR-SMDL}
 There is a randomized $\nicefrac{1}{2}$-approximation for DR-SMDL. 
\end{theorem}

We start by describing the modified algorithm, for a formal description see Algorithm \ref{algo:DLDoubleGreedy}. 
To avoid complicating it, read $\frac{0}{0} = 0$ in the algorithm.  
In contrast to the algorithm for SMBIL we presented in Section \ref{sec:SMBIL-Algo}, we will 
only add or delete one element at a time (which corresponds to adding or subtracting 1 for SMBIL). 
Given a distributive lattice $\cD(P)$ on a poset $P = (\cN, \preceq)$, we fix a linear extension of $P$ to obtain a total ordering
$x_1, \ldots, x_n$ of $\cN$. Starting from $A= \emptyset$ and $B = \cN$, 
we always consider a pair $(x_{min}, x_{max})$ consisting of the smallest available element in the ordering and a maximal
element $x_{max}$ such that the pair is comparable in $P$. Note that $x_{min} = x_{max}$ is permitted.   
Then, we randomly choose to add $x_{min}$ to $A$ or delete $x_{max}$ from $B$ with probabilities proportional to the increase in $f$ and 
remove the chosen element. This continues until $A = B$. 
DR-submodularity is essential here.  Without it, the example presented in Figure \ref{fig:BadLatticeOneGreedy} 
with the linear extension $c, d, a, b$ shows that Algorithm \ref{algo:DLDoubleGreedy} 
would perform arbitrarily badly even for the integer lattice. 

\begin{algorithm2e}
\DontPrintSemicolon
\textbf{Input:} A nonnegative DR-submodular function $f$ on a distributive lattice $\cD(P)$\;
\textbf{Output:} A set $A \in \cD(P)$\;

 $A= \emptyset$, $B = \cN$, $x_1, \ldots, x_n$ a linear extension of $P$, $i = 1$\; 
 $k = 1$, $j  = \max \{ j' : x_k \preceq x_{j'} \}$\; 
 \While {$A \not = B$}{
 	$a_i = f(A + x_{k}) - f(A)$\;
	$b_i = f(B - x_{j}) - f(B)$ \;
	With probability $\frac{\max\{a_i, 0\}}{\max\{a_i, 0\} + \max\{b_i, 0\}}$ set $A = A + x_{k}$, $k = k + 1$ 
	and find new $j  = \max \{ j' : x_k \preceq x_{j'} \}$. \;
	With complementary probability set $B = B - x_{j}$, find new $j$ for k.\;	
}
    \Return A\;
 \caption{Greedy with Matroid Constraint}\label{algo:DLDoubleGreedy}
\end{algorithm2e}

Adapting the arguments of \cite{DoubleGreedy} to our setting is straightforward. 
We start by showing that for at least one of the two options, $f$ does not decrease.  
\begin{lemma}\label{lemma:DLpositivechange}
In every iteration of Algorithm \ref{algo:DLDoubleGreedy}, $a_i + b_i \ge 0$. 
\end{lemma}
Since $A \subseteq B - x_j$ and $x_k \le x_j$, the above expression is just a reformulation of DR-submodularity and thus the Lemma follows. 
We fix an optimal solution $OPT$ and let $A_i$ and $B_i$ be random variables denoting the sets $A$ and $B$ of elements after the i-th iteration. 
We show that for DR-submodular functions, we can get a good bound on the expected decrease in $f(OPT_i)$, where the random variable 
$OPT_i$ is defined as $(OPT \cup A_i) \cap  B_i$. 

\begin{lemma}
For every iteration $i$, we have $E[f(OPT_{i - 1}) - f(OPT_i)] \le \frac{1}{2} \cdot E[f(A_i) - f(A_{i - 1}) + f(B_i) - f(B_{i -1})]$. 
\end{lemma}
Then, as presented in Section \ref{sec:SMBIL-Algo}, by summing over all iterations and collapsing the resulting telescopic sum, 
we obtain the proof of Theorem \ref{thm:AlgoDR-SMDL}. 
It remains to show the lemma. 
\begin{proof}
The proof works completely analogous to \cite{DoubleGreedy}. We will give the details here and omit them for similar results. 
It suffices to prove the inequality conditioned on $A_{i - 1} = S_{i - 1}$ for a set $S_i$ consisting of elements
$x_{min}$ and a corresponding series of pairs $(x_{min}, x_{max})$ when the probability that this
series of pairs leads to $A_i = S_i$ is non-zero. We now fix such an event.
Then, $B_{i-1}$, $OPT_{i - 1}$, $a_i, b_i$ and $x_{min},x_{max}$ become constants. 
By Lemma \ref{lemma:DLpositivechange}, $a_i$ or $b_i$ must be nonnegative. 
Therefore, it suffices to consider the following cases. 

Suppose $b_i \le 0$ and thus $a_i \ge 0$. Then, $A_i = A_{i - 1} + x_{min}$ and $OPT_i = (OPT \cup (A_{i -1} + x_{min}))\cap B_{i - 1}$ and 
$x_{min}, x_{max} \in B_i$.
If $x_{min} \in OPT$, we have $0 = f(OPT_{i - 1}) - f(OPT_i) \le \frac{a_i}{2}$.   
If $x_{min} \not \in OPT$, then 
\begin{align*}
 f(OPT_{i - 1}) - f(OPT_i) &= f(OPT_{i - 1}) - f(OPT_{i - 1} + x_{min}) \\ 
 &\le f(B_{i -1} - x_{max}) - f(B_{i - 1}) = b_i 
\le 0  \le \frac{a_i}{2}
\end{align*}

by DR-submodularity. 
The case where $a_i \le 0$ is analogous.

Therefore, suppose $a_i > 0$ and $b_i > 0$. 
Then, 
\begin{align*}
&E[f(A_i) - f(A_{i - 1}) + f(B_i) - f(B_{i -1})] \\
= &\frac{a_i}{a_i + b_i}(f(A_{i - 1} \!+\! x_{min})\!-\!f(A_{i - 1})) + 
						   \frac{b_i}{a_i + b_i}(f(B_{i - 1} \!-\! x_{max})\!-\!  f(B_{i - 1})) \\
= &\frac{a_i^2 + b_i^2}{a_i + b_i}			\numberthis \label{eq:DLexpectedSets}
\end{align*}						  

On the other hand, 
\begin{align*}
&E[f(OPT_{i - 1}) - f(OPT_i)] \\
= &\frac{a_i}{a_i + b_i}(f(OPT_{i - 1}) \!-\!f(OPT_{i - 1} + x_{min})) \\ 
						   &+ \frac{b_i}{a_i + b_i}(f(OPT_{i - 1} \!-\!  f(OPT_{i - 1} - x_{max}))\\
\le& \frac{a_ib_i}{a_i + b_i} 				\numberthis \label{eq:DLexpectedOPT}
\end{align*}

We still need to explain the final inequality. 
If $x_{max} \not \in OPT$, the second term on the left hand side equals zero. For the first term , we use DR-submodularity to obtain
$f(OPT_{i - 1}) \!-\!f(OPT_{i - 1} + x_{min}) \le f(B_{i-1} - x_{max}) - f(B_{i - 1}) = b_i$. 

Otherwise, $x_{min} \in OPT$ as OPT is an ideal and thus, the first term is zero. Moreover, 
DR-submodularity implies   
$f(OPT_{i - 1}) \!-\!  f(OPT_{i - 1} - x_{max}) \le f(A_{i - 1} + x_{min}) - f(A_{i - 1})  = a_i$.

Now the inequalities \eqref{eq:DLexpectedOPT} and \eqref{eq:DLexpectedSets} together with the fact that 
$\frac{a_i b_i}{a_i + b_i} \le \frac{1}{2}\cdot \frac{a_i^2 + b_i^2}{a_i + b_i}$ yield the desired result. 
\end{proof}

The deterministic version from \cite{DoubleGreedy} which yields a $\frac{1}{3}$-approximation can also be generalized to this case. 
In contrast to the non-DR-setting in Section \ref{sec:SMBIL-Algo}, the proofs from \cite{DoubleGreedy} by 
Buchbinder et al.\ can be directly generalized together with the ideas used in this section. Therefore, we will not present them here. 

As shown in \cite{DerandomizedDoubleGreedy}, there is also a deterministic $\nicefrac{1}{2}$-approximation based on the randomized algorithm. 
The rough idea is the following. They maintain a set of possible ``states'' of sets with corresponding probability. For example, 
in the beginning, there is exactly one possible state: $A = \emptyset$ and $B= \cN$.
In iteration $i$, they consider element $u_i$ and for each current state $(A, B)$, they assign probabilities to adding $u_i$ to $A$ or deleting
it from $B$. This way, they build the set of states for iteration $i$. For example, after one iteration, 
the set of states would then be $(u_1, \cN)$ with probability $p$ and $(\emptyset, \cN \setminus u_1)$ with probability $1-p$. 
The groundbreaking idea in this paper is how to assign these probabilities. There are two aspects to this. 
First, the set of states should stay small, in particular, doubling its size in every iteration is not an option. Second, the expected decrease in 
the objects $OPT(A, B)$ corresponding to states $(A, B)$ should be small. 
Both these properties are achieved by determining the probabilities using a linear program. The first property follows from the fact 
that an extreme point of their linear program must have small support. The second property is enforced by inequalities in the LP. 

Their strategy can also be generalized to the DR-submodular setting on distributive lattices. But, there is one additional technical issue. 
In \cite{DerandomizedDoubleGreedy}, the element that is to be added or deleted is the same for all states in an iteration $i$. 
For our setting, each state $(A, B)$ also comes with its own pair $x_{min} \preceq x_{max}$. 
However, the objects $OPT(A,B)$ used for the analysis also depend on the state and for a fixed state we still obtain the necessary properties 
like $OPT(A, B) \subset B\setminus x_{max}$.
Therefore, we can use DR-submodularity in the same way as \cite{DerandomizedDoubleGreedy} 
to obtain a deterministic $\nicefrac{1}{2}$-approximation for SMDL. Since the technical details are very similar to the ones already provided we will
not present them. 

It seems natural to ask whether the approach presented in Section \ref{sec:SMBIL-Algo} for SMBIL, i.e.\ for non-DR functions, 
could also be generalized to distributive lattices. 
This would mean choosing a best join-irreducible element of the lattice to add to $A$ and a best complement of a meet-irreducible element to delete from $B$. 
Which elements would be eligible to be chosen would be subject to discussion. 
But, using such a strategy, it is possible to end up in situations where there is no possibility of changing 
$A$ or $B$ as described above without decreasing their values w.r.t.\ $f$. 

\section{DR-submodularity on Distributive Lattices: Hardness}\label{sec:SM-DR-DL-hardness}
As the previous algorithms maximizing DR-submodular functions on the distributive lattice were direct generalizations of algorithms for the Boolean lattice, 
one could assume that there is no significant structural difference. However, in this section we
consider the problem of maximizing a monotone DR-submodular function $f: \cD(P) \rightarrow \R_+$ on the distributive 
lattice subject to a knapsack constraint with binary knapsack weights. Formally, 
\begin{align*}
 \max\{f(S): w(S) \le B\} \text{ for } w: \cN \rightarrow \{0, 1\} \numberthis \label{eq:KnapsackDL}
\end{align*}

This problem is essentially as hard to approximate as the Densest $k$-subhypergraph problem and thus, in particular there is no constant factor approximation
under some complexity assumptions. 
In contrast, for monotone functions on the Boolean lattice, there is a $(1- \nicefrac{1}{e})$-approximation for one (arbitrary) knapsack constraint \cite{SviridenkoKnapsack} 
and a  $(1- \nicefrac{1}{e}- \epsilon)$-approximation(\cite{KulikMonotoneKnapsack}) for $\mathcal{O}(1)$ Knapsack constraints. 
Moreover, \cite{SomaYoshi16} generalized these results to the integer lattice and obtained a guarantee of $(1- 1/e- \epsilon)$ for monotone 
DR-submodular functions and positive knapsack coefficients and a 
pseudopolynomial $(1 - \nicefrac{1}{e})$-approximation for monotone submodular functions 
\cite{Inaba:KnapsackIL}. Note that due to  monotonicity of $f$, the result directly generalizes to nonnegative knapsack values. 

We start by introducing the densest $k$-subhypergraph problem (D$k$SH). Given a hypergraph $H = (V, E)$ and $k \in \N$, 
the goal is to choose $k$ vertices such that the number of induced hyperedges is maximized. 
One special case is the densest $k$-subgraph problem for which, although it is believed to be hard, 
it is only known that no PTAS exists if NP has no subexponential algorithms \cite{Khot04DkS}. 
However, Hajiaghayi at al.\ showed in \cite{HajiSubhypergraph} that D$k$SH is hard to approximate within a factor of $2^{(\log n)^\delta}$ 
for some $\delta$ under the assumption that 
$3-SAT \not \in DTIME(2^{n^{3/4 + \epsilon}})$, i.e.\ $3-SAT$ cannot be solved in time $2^{n^{3/4 + \epsilon}}$. 

\begin{theorem}
 \label{thm:SubmaxKnapsackHardness}
 If there is an $\alpha(n)$-approximation for maximizing a monotone DR-submodular function on a distributive lattice subject to a $0$-$1$-Knapsack constraint, 
 then there is a $2\alpha(n'(m' + 1))$-approximation for D$k$SH, where $n$ and $(n', m')$ denote the respective input sizes and 
 $\alpha$ is positive and non-decreasing. 
\end{theorem}

\begin{proof}
Let $H = (V, E)$, $k \in \N$ be the input of D$k$SH.  
Our goal is the construction of a submodular maximization instance. 
We construct a poset $P = (\cN, \preceq)$ by introducing an element $v$ for every vertex $v \in V$ and 
$k$ elements $e^1, \ldots, e^k$ for every hyperedge $e \in E$. 
Slightly abusing notation, we are also going to denote by $V$ the subset of elements in $\cN$ corresponding to vertices in $H$.  
Then, $v \le e^i \ \forall i \le k, v \in V$ if $v \in e$ in $H$. 
Moreover, $f(S) = |S|$ which clearly is DR-submodular and monotone and $n = n' + km' \le n'(1+m')$. 
The knapsack constraint is given by asking that $|S\cap V| \le k$, i.e.\ $B = k$ and $w(v) = 1$ for $v \in V$ and $w(e) = 0$ otherwise. 

Any ideal $I \in \cD(P)$ corresponds to the set of vertices $I \cap V$ and some hyperedges induced by this set.
On the other hand, any induced subhypergraph in $H$ can be represented (not uniquely) by an ideal. 
In particular, there is a one-to-one correspondence between induced subhypergraphs and ideals $I$ where $I \setminus V$ is maximal. 

Let $S$ be the result of an $\alpha(n)$-approximation for our problem. Since the knapsack constraint only 
affects the elements in $V$, we assume w.l.o.g.\ that $S$ contains either all copies $e^i$ for an edge $e$ or none of them and that $|S\cap V| = k$. 
Therefore, $S\cap V$ corresponds to a feasible solution for D$k$SH inducing $\frac{|S| - k}{k} =: \beta$ edges.
We can assume that $\beta \ge 1$. If not, we construct a solution fulfilling that easily.  

Let $\mathcal{R}$ denote the optimal value for D$k$SH. Then, the optimal value for our submodular maximization instance is 
$k(1 + \mathcal{R})$. 
Since $S$ was obtained by an $\alpha(n)$-approximation,  
$f(S) = |S| \ge \frac{1}{\alpha(n)}k(1 + \mathcal{R})$. 
Combining the above, we obtain 
$$ k(1 + \mathcal{R}) \le \alpha(n)|S| = \alpha(n)(\beta + 1)k \le \alpha(n)2\beta k.$$
Therefore,
$ 2\alpha(n'(m' + 1))\beta \ge 2\alpha(n)\beta \ge \mathcal{R}$ and thus, we have a $2\alpha(n'(m' + 1))$ approximation for D$k$SH.  
\end{proof}

We could refine the constant of 2 in the previous analysis: 
For any constant $c$ we can find a solution of D$k$SH 
of value at least $c$ by enumeration. (Unless the optimum is less than $c$, then we solve D$k$S optimally.)
If our algorithm performs worse, we take the solution of size $c$ instead and therefore can assume 
$\beta \ge c$. This leads to a $(1 + \frac{1}{c})\alpha(n'(m' + 1))$-approximation in the above analysis.

Using the above theorem we obtain: 

\begin{corollary}
The submodular maximization problem defined in \eqref{eq:KnapsackDL} is hard to approximate within a factor  
of $2^{(\log (n^{1/2} - 1))^\delta - 1}$ for every constant $\delta > 0$ under the assumption that 
$3-SAT \not \in DTIME(2^{n^{3/4 + \epsilon}})$ for every constant $\epsilon > 0$. 
\end{corollary}

\begin{proof}
We need to examine the hardness results for D$k$HS from \cite{HajiSubhypergraph} a little closer. The authors reduce from a problem presented in \cite{FeigeKogan04}, the so-called 
Maximum Balanced Complete Bipartite Subgraph Problem. However, they omit the fact that the input graph for that problem is not just bipartite but balanced as well. 
This ensures that the complexity result for D$k$SH is true even if the number of hyperedges is equal to the number of nodes. 
Moreover, \cite{FeigeKogan04} state their result for all constants $\delta$, not just for some constant as cited and used in \cite{HajiSubhypergraph}, 
which is why we also use this more general formulation here.  

Therefore, we can now apply Theorem \ref{thm:SubmaxKnapsackHardness} for a hypergraph with $n' = m'$  and use the same reduction as above. 
Therefore, $n \le n'^2 + n' \le (n' + 1)^2$ and thus
$n' \ge \sqrt(n) - 1$. Therefore, by setting $\alpha(n) = 2^{\log(n^{1/2} - 1)^\delta - 1}$, we obtain our result. 
\end{proof}


\bibliographystyle{plain}
\bibliography{SubMaxFullLiterature.bib}

\end{document}